\documentclass{amsart}

\usepackage[dvips]{graphicx}
\usepackage{bbm}

\usepackage[latin1]{inputenc}
\usepackage{epsfig}
\usepackage{subfig}
\usepackage{fixltx2e}
\graphicspath{{figures/}}

\usepackage{color}
\newtheorem{proposition}{Proposition}

\usepackage{epstopdf}
\newtheorem{lemma}{Lemma}

\def\R{{\mathbb R}}
\def\P{{\mathbb P}}
\def\E{{\mathbb E}}

\def\cal{\mathcal}

\def\etal{{\em et al. }}

\newcommand\ind[1]{\mathbbm{1}_{\left\{#1\right\}}}

\usepackage{cite}
\author[C. Fricker]{Christine Fricker}

\author[Ph. Robert]{Philippe Robert}

\author[J. Roberts]{James Roberts}

\address[C. Fricker, Ph. Robert, J. Roberts]{INRIA Paris --- Rocquencourt,  Domaine de Voluceau, 78153 Le Chesnay, France.}

\email{Christine Fricker@inria.fr}

\email{Philippe.Robert@inria.fr}
\urladdr{http://www-rocq.inria.fr/\string~robert}

\email{James.Roberts@inria.fr}

\begin{document}

\title{A versatile and accurate approximation for LRU cache performance }

%\author{\IEEEauthorblockN{Christine Fricker, Philippe Robert, James Roberts} \IEEEauthorblockA{INRIA, France\\ \{Christine.Fricker,Philippe.Robert,James.Roberts\}@inria.fr }}


%\author{\IEEEauthorblockN{Nelson Antunes\IEEEauthorrefmark{1},
%Christine Fricker\IEEEauthorrefmark{2},
%Philippe Robert\IEEEauthorrefmark{2}, 
%James Roberts\IEEEauthorrefmark{2} }
%\IEEEauthorblockA{\IEEEauthorrefmark{1}  University of the Algarve, Portugal (gomes.antunes@gmail.com)},
%\IEEEauthorblockA{\IEEEauthorrefmark{2} INRIA, France (\{christine.fricker,philippe.robert,james.roberts\}@inria.fr )}
%}

\begin{abstract}
%\boldmath
In a 2002 paper, Che and co-authors proposed a simple approach for estimating the hit rates of a cache operating the least recently used (LRU) replacement policy. The approximation proves remarkably accurate and is applicable to quite general distributions of object popularity. This paper provides a mathematical explanation for the success of the approximation, notably in configurations where the intuitive arguments of Che \etal clearly do not apply. The approximation is particularly useful in evaluating the performance of current proposals for an information centric network where other approaches fail due to the very large populations of cacheable objects to be taken into account and to their complex popularity law, resulting from the mix of different content types and the filtering effect induced by the lower layers in a cache hierarchy. 
\end{abstract}

\maketitle

%%%%%%%%%%%%%%%%%%%%%%%%%%%%%%%%%%%%%%%%%%%%%%%%%%%%%%%%%%%%

\section{Introduction}

The investigation of so-called information-centric networking (ICN) architectures is bringing renewed interest in the performance of caching. It is particularly important to understand the potential for trading off bandwidth for memory by implementing a network of caches and to develop tools that enable the optimization of such a network. The ICN application places particularly stringent requirements on evaluation tools since the population of content items available via the Internet is immense and caches are required to store content of diverse types, each type being distinguished by its peculiar popularity characteristics. 

In recent work on cache performance in the context of ICN \cite{FRRS12}, we applied a tool from the literature that was particularly well adapted to requirements. This is an approximation for evaluating the hit rates of a cache under the least recently used (LRU) replacement policy proposed by Che, Tung and Wang in a 2002 paper  \cite{CTW02}.  The ``Che approximation'' proved extremely accurate, even in conditions where the authors' intuitive arguments were clearly not justified. The objective of the present paper is to provide more rigorous mathematical arguments allowing the scope of the approximation to be more clearly defined.  

The Che approximation applies to the following model. Users request items from a population of $N$ objects, first testing to see if the object is present in a cache of capacity $C$. If the object is present it is returned to the user. If not, it is obtained from some other source and copied to the cache as it is returned to the user. This object replaces the one that was least recently requested. The probability a request is for object $n$, for $1\le n \le N$, is proportional to some popularity $q(n)$, independently of all past requests. 

The hit rate $h(n)$ for object $n$, i.e., the probability this object is present in the cache, is approximated by 
$$h(n) \approx 1-e^{-q(n)t_C}$$
where $t_C$ is the unique root of the equation
\begin{equation*}
 \sum_{n=1}^N (1-e^{-q(n)t})= C.
 \end{equation*}

%The Che approximation is versatile since the popularity law $q(\cdot)$ can be quite general. This is useful for ICN evaluations since the popularity is a composite
%advantages of Che: generality, mixtures, easy numerical calculations , can apply to non-symmetric networks (different cache sizes, different demands,...), like Rosenweig \etal  \cite{RKT2010}. Further discussion on the Che approx is provided in \cite{LCS06}.

\noindent \emph{Related work}

There is clearly a huge amount of related work on the performance of caching. We restrict ourselves here to a discussion of the papers most relevant to our work. Cited papers \cite{DT90} and \cite{Jelenkovic99} may be consulted for summaries of significant early work. 

Dan and Towsley \cite{DT90} derived an iterative algorithm for calculating approximate hit rates for a cache of size $C$ using the hit rates for a cache of size $C-1$. Complexity is $O(CN)$ which can be prohibitive in ICN applications where $N$ and $C$ are very large.  Dan and Towsley also propose a scheme comparable in complexity to the Che approximation for computing hit rates under FIFO replacement (i.e., the object  replaced is the one that has been in the cache the longest). Recent work by Rosenweig \etal applies the LRU algorithm of \cite{DT90} to analyse general networks of caches \cite{RKT2010}. %Their approach could alternatively use the Che algorithm.

Jelenkovic provides closed-form asymptotic hit rate estimates for particular choices for popularities $q(n)$  \cite{Jelenkovic99}. These are namely, a generic light-tailed law, $q(n)=e^{-\lambda n^{\beta}}$ for $n>0$, where $\lambda$ and $\beta$ are positive constants, and a generalized Zipf law, $q(n)=1/n^{\alpha}$ for $n\ge 0$, with $\alpha>1$. The latter law with $\alpha < 1$ and $\alpha =1$ is considered by Jelenkovic \etal in \cite{JKR05} where it is shown that hit rates can be expressed in terms of a parameter defined as the root of a certain equation.  The main disadvantage is that derived formulas are only applicable to particular popularity laws.

%The Che approximation is superior to prior work in that it applies to any popularity law and is sufficiently computationally efficient to be used for large populations and cache sizes appropriate for the ICN application.
\vspace{2mm}
In the following, we first discuss the considered traffic model in Section \ref{sec:trafficmodel}  before presenting the Che approximation in detail in Section \ref{sec:cheapprox}. Reasons for its remarkable accuracy are elucidated in Section \ref{sec:whyitworks} while Section \ref{sec:largecache} derives explicit results for the special case of Zipf popularity. An approximation for random replacement similar to the Che approximation is derived in Section \ref{sec:randomche} and used in Section \ref{sec:application} in an ICN application.

%Jelenkovic and co-authors have 
% related work: Jelenkovic \cite{Jelenkovic99, JKR05, JK08}, (Luca \cite{CGMP11}), Dan and Towsley .cite{DT90}, Che \cite{CTW02}, Laoutaris \cite{LCS06}...

 \section{Traffic model}
 \label{sec:trafficmodel}
 
We recall the independent reference model (IRM), discuss the nature of the popularity law $q(\cdot)$ and argue that the IRM is appropriate for modelling an information-centric network.

\subsection{The independent reference model}
 
Requests for objects occur in an infinite sequence where the object index required on the $i^{th}$ request,  for $i>0$, is an independent random variable on $\{1,2,\ldots,N\}$ with a common probability distribution. Specifically, the probability the required object has index $n$ is proportional to $q(n)$, for $1\le n \le N$. We refer to $q(\cdot)$ as the \emph{popularity law}. %This assumption ensures the cache state specifying the identity of the object in each cache position $i$, for $1\le i \le C$, is Markovian for both LRU and random replacement policies.
 
An alternative description of this stochastic is as follows (see Fill and Holst~\cite{FillHolst}, for example). Let $(\tau_n)$ be a sequence of independent exponential random variables with respective rates $(q(n))$. At any arbitrary instant, the next object to be requested is $n_0$ where $n_0$ is the index such that $\tau_{n_0}$ is the minimum of the $(\tau_n)$. 

\subsection{Popularity laws}
Cache performance depends crucially on the popularity law $q(\cdot)$. It is usual to order objects in order of deceasing popularity such that $q(1)\ge q(2) \ge \ldots \ge q(N)$. With this convention, the most frequently observed popularity law is a generalized Zipf law: $q(n)=1/n^\alpha$ with $\alpha>0$. 

Examples of content types with reported Zipf law behaviour are web pages \cite{Breslau99},  \cite{Mahanti00}, files shared using BitTorrent~\cite{FRRS12}, YouTube documents \cite{Gill07},  \cite{Cha07},  video on demand movies  \cite{Yu06}. 
The reason why content popularity follows the Zipf law remains unclear though the discussion by Mitzenmacher on generative models for power laws provides some possible explanations \cite{Mitzenmacher03}. The estimated value of $\alpha$ is often around 0.8 though cases with $\alpha>1$ have been observed. 

In fact, agreement between observations and the Zipf law is not always entirely convincing. Sometimes the popularity law has a lighter tail where the least popular objects are very unlikely to be requested. As a simple example of a light-tailed law, we consider geometric popularity: $q(n)=\rho^n$. This choice is not based on any measurement results but is made rather with the intention of stressing the Che approximation: the approximation is in fact exact for a uniform probability law, $q(n)=$constant; it is more likely to fail as the law becomes more accentuated, as with the geometric law.

While a Zipf  law with suitable $\alpha$ might be a reasonable representation for a homogeneous set of content, caches in the Internet must be designed for a traffic mix. The popularity law should reflect this mix by weighting single type popularities by the proportion of requests due to that type. It is necessary also to account for significant differences in the size of objects of different types (see Section \ref{sec:varsize}). 

\begin{figure}
\centering
\vspace{-13mm}
\hspace{-15mm}
{%\large
{\input{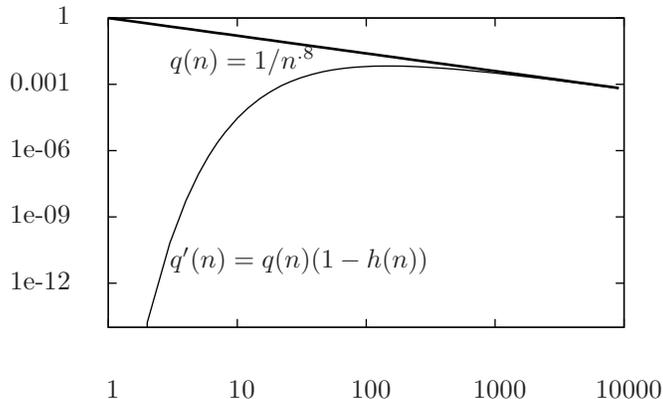}}
}
\vspace{-10mm}
\caption{A filtered popularity law: $q'(\cdot)$ is the popularity law for requests overflowing a size 1000 cache  when $q(\cdot)$ is Zipf(0.8) and $N=10000$.}
\label{fig:filter}
\end{figure}

In a network, there is typically a hierarchy where higher layer caches receive requests only for objects that are not found in lower layers. The popularity law is thus a filtered version of the initial law. Figure \ref{fig:filter} shows how the popularity law at a second layer cache is deformed by the first layer. 

Fortunately, the Che approximation is sufficiently versatile to account for both composite popularity laws and filtering \cite{FRRS12}. Note that it is not necessary to order objects in decreasing order of popularity or to normalize the $q(n)$ and we do not make these assumptions in the following analysis.

\subsection{Validity of the independent reference model}
The independent reference model is a convenient abstraction that allows analytical modelling where a more accurate representation of reality would be intractable. It is important to understand the limitations of this model.

Adopting the IRM is to assume popularity does not change. This is clearly not true for content where, not only the popularity of a given object, but also the catalogue of available objects change over time.  This is a manifestation of \emph{temporal locality}. The IRM may still be considered acceptable if popularity variations are slow compared to the time scale of cache churn. Non-stationarity has a greater impact on the accuracy of measurements of popularity. For instance, statistics gathered over a period of weeks will hardly be representative of the popularity of a catch-up TV show that is only on-line for a few days.

A second cause of error in predicting popularity laws is \emph{spatial locality}.  Most types of content will have strong regional bias with, for instance, a network in France observing high popularity for movies in French. However, this does not so much invalidate the IRM assumption as require more care in specifying popularity laws.

The IRM is reasonable when content requests are generated independently by a large population of users. This is not the case, however, for requests seen by caches at second and higher layers in a hierarchy. The request process overflowing from lower layer caches is correlated. The IRM is nevertheless reasonable if a higher layer cache receives the aggregation of independent, low intensity overflows from many low layer cache instances. Moreover, even for the simple tandem of two caches considered by Jenekovic and Kang \cite{JK08}, the impact of correlation on hit rates was shown to be slight.

We conclude from the above that the IRM is a reasonable basis for evaluating cache performance, as long as care is taken in specifying the popularity law. %This motivates our interest in demonstrating the general applicability of the Che approximation that relies on this fundamental assumption.

%%%%%%%%%%%%%%%%%%%%%%%%%%%%%%%%%%%%%%%%%%%%%%%%%%%%%%%%%%%%%%%

\section{The Che approximation}
\label{sec:cheapprox}
We present the Che approximation and demonstrate its accuracy in comparison to the results of simulation.

\subsection{A characteristic time}
\label{criticaltime}
%Under LRU replacement in a cache of size $C$, a hit for object $n$ occurs if there have been fewer than $C$ requests for objects other than $n$ since the last time object $n$ was requested. Recalling that object $i$ request inter-arrivals are exponential with parameter $\tau_i$, the hit rate for object $n$ can therefore be expressed: 
%\[
%\sum_{i\not=n}\ind{\tau_i<\tau_n}<C.
%\]
Consider a cache with capacity for $C$ objects under the LRU replacement policy. We introduce the following random variables, for $t\geq 0$, 
\[
X_n(t)=\sum_{i=1, i\ne n}^{N}\ind{\tau_i<t}
\]
and
 $$T_C(n)=\inf\{t>0: X_n(t)=C\}.$$
$X_n(t)$ is the number of different objects requested up to time $t$, excluding object $n$, and $T_C(n)$ is the time at which exactly $C$ different objects, other than $n$, have been requested. 

Without loss of generality, suppose a request for object $n$ occurs at time 0. The next request for object $n$ will be a hit if fewer than $C$ other objects are requested in $(0,\tau_n)$ (recall that $\tau_i$ is the generic, exponentially distributed inter-request interval  for object $i$). In other words, there is a hit if $X_n(\tau_n)<C$. Now,  $$\{X_n(\tau_n)<C\} =\{T_C(n)>\tau_n\}$$ so that, 
$$h(n)=\P(T_C(n)>\tau_n)=\E\left(1{-}e^{-q(n)T_C(n)}\right).$$

Since at $T_C(n)$ there are exactly $C$ objects in the cache, we have 
$$C=\sum_{i=1,i\ne n}^{N}\ind{\tau_i<T_C(n)},$$
and taking expectations,
\begin{equation*}
%\label{eq:hitC}
C=\sum_{i=1,i\ne n}^{N}\E\left(1-e^{-q(i)T_C(n)}\right).
\end{equation*}

A first approximation of Che \etal \cite{CTW02} is to assume for large $C$ that the $T_C(n)$ are nearly deterministic. They replace random variable $T_C(n)$ by the constant $t_C(n)$ that solves
\begin{equation*}
%\label{eq:hitC}
C=\sum_{i=1,i\ne n}^{N}\left(1-e^{-q(i)t}\right)
\end{equation*}
and approximate the hit rates by
$$h(n)=1-e^{-q(n) t_C(n)}.$$

A second approximation in \cite{CTW02} is to assume  $t_C(n)=t_C$ for $1 \le n \le N$ where $t_C$ solves $$C=\sum_{1\le i \le N}(1-e^{-q(i)t}).$$ This is arguably reasonable when individual popularities $q(n)$ are small relative to the sum $\sum_n q(n)$. Having verified numerically that this second approximation is generally very accurate, we adopt it for the remainder of the paper. This simplifies notation but it should be noted that our analysis could readily be adapted to preserve the dependence on $n$.

In summary, the Che approximation considered in this paper is as follows. Let $t_C$  be the unique root of the equation
\begin{equation}
\label{eq:hitC}
C=\sum_{i=1}^{N}\left(1-e^{-q(i)t}\right).
\end{equation}
The hit rate $h(n)$ for object $n$, for $1 \le n \le N$, is then 
\begin{equation}
\label{eq:hitrates}
h(n) =1-e^{-q(n)t_C}.
\end{equation}
Che \etal  refer to $t_C$ as the ``characteristic time'' of the cache.

%The right hand side of (\ref{eq:hitC}) is increasing in $t$, ensuring the unicity of $t_C$. 
%When $N$ is very large (O($10^9$ say for the population of web pages), the summation in (\ref{eq:hitC}) can be simplified by averaging small $q(n)$ over .... 

Random variables $X(t)$ and $T_C$ are defined as above without excluding object $n$. In particular, $X(t) = \sum_{1\le i \le N} \ind{\tau_i<t}$. Since the Bernoulli events in the summation are independent, we readily derive the mean and variance of $X(t)$:
\begin{eqnarray}
m(t)&=&\sum_{i=1}^N (1-e^{-q(i)t}), \label{eq:m(t)} \\
\sigma(t)^2&=&\sum_{i=1}^N e^{-q(i)t} (1-e^{-q(i)t}). \label{eq:sigma2(t)}
\end{eqnarray}

 \subsection{Variable sized objects}
 \label{sec:varsize}
 
So far we have assumed the cache capacity is measured in objects. In reality content objects have different sizes and cache capacity is more reasonably measured in bytes. Suppose object $n$ has size $\theta(n)$. Since the cache is intended to store a very large number of objects, $\theta(n)\ll C$ so that we can reasonably ignore boundary effects and adapt the Che approximation by replacing (\ref{eq:hitC}) with
\begin{equation}
\label{eq:hitCvar}
C=\sum_{i=1}^{N}\left(1-e^{-q(i)t}\right)\theta(i).
\end{equation}
The hit rates are still given by (\ref{eq:hitrates}).

An alternative way to account for variable size objects is to assume they are divided into constant sized chunks. This is the principle of content-oriented Internet architectures like CCN \cite{JSTP09}. Let $\theta(n)$ be given in chunks and assume all chunks inherit the popularity $q(n)$ of their parent object. Applying the Che approximation to chunks, it is easy to see that equation (\ref{eq:hitC}) for $t_C$ is then precisely the same as (\ref{eq:hitCvar}).

For some types of objects, it may be that chunks of the same object have different popularities (e.g., the first chunks of a video will be viewed more often than the last chunks). It makes more sense in this case to directly postulate a popularity law for chunks rather than objects. 

We conclude that the Che approximation presented in Section \ref{criticaltime} is appropriate also for variable size objects. It is not necessary to complicate the analysis in the next sections by introducing the size $\theta(n)$ (although to use (\ref{eq:hitCvar}) might be useful in practice as in \cite{FRRS12}).

\begin{figure}
  \centering
\  \vspace{-3mm}\hspace{-2cm}
  \subfloat[$N=10^4$, Zipf(0.8), objects 1, 10, 100, 1000] {\label{fig:hitzipf8} \includegraphics[width=80mm,height=50mm]{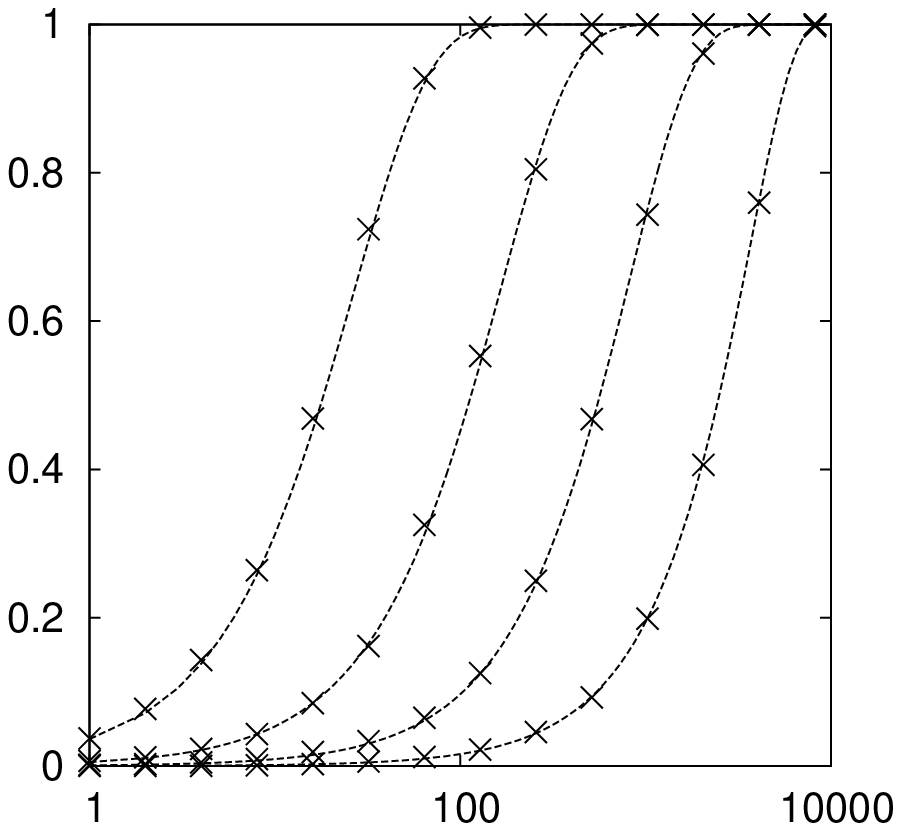} }
%\hspace{-15mm}     
  \subfloat[$N=10^4$, Zipf(1.2), objects 1, 10, 100, 1000] {\label{fig:hitzipf12} \includegraphics[width=80mm,height=50mm]{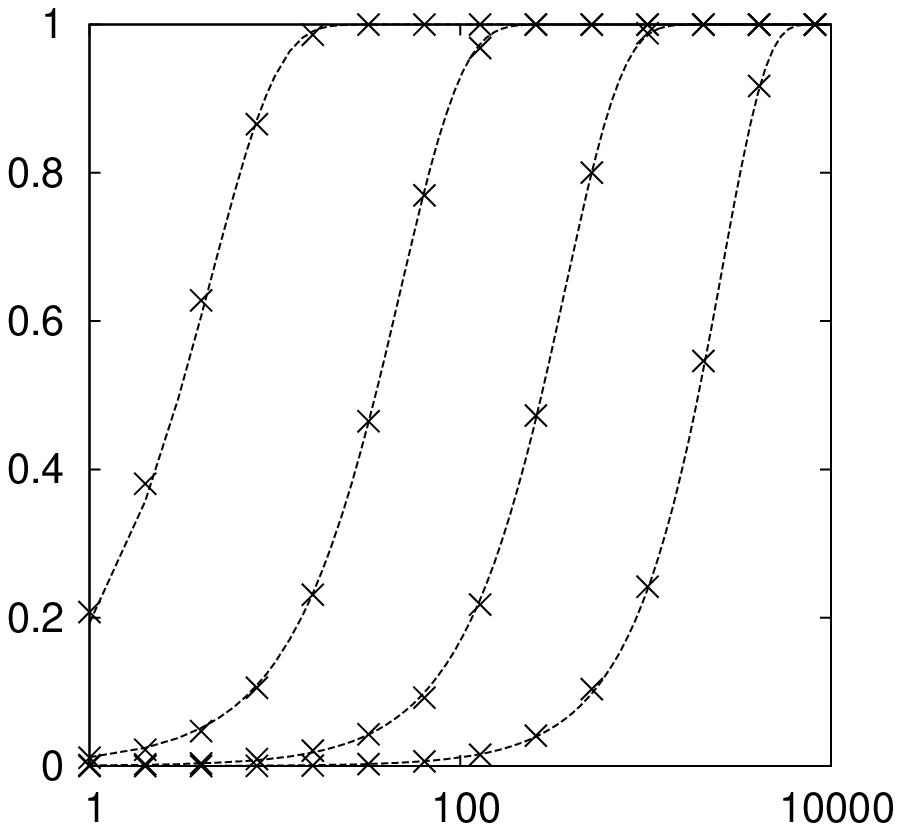} } \\
% \hspace{-15mm}     
\subfloat[$N=100$, Geo(0.9),  objects 1, 4, 16, 64] {\label{fig:hitgeo9} \includegraphics[width=80mm,height=50mm]{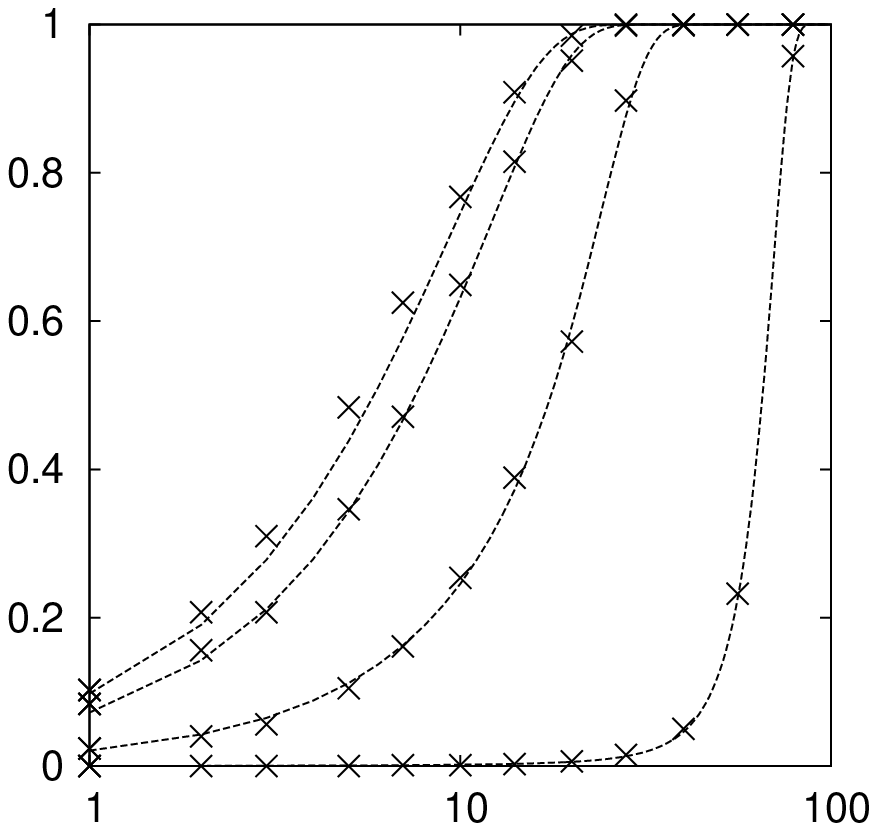} }
   \caption{Hit rate against cache size for selected objects, LRU replacement}
  \label{fig:approx}
  \vspace{-3mm}
\end{figure}

\subsection{Accuracy}
The accuracy of the approximation is typified by the results shown in Figure \ref{fig:approx}.  Figures \ref{fig:hitzipf8} and \ref{fig:hitzipf12} plot the hit rates of objects ranked 1, 10, 100 and 1000 from a population of 10000, assuming Zipf popularity with $\alpha=0.8$ and $\alpha=1.2$, respectively.  The crosses are the results of simulations with sufficiently long runs to ensure their high accuracy. The lines are derived from the Che approximation. Agreement is perfect, for all practical purposes.  

Figure \ref{fig:hitgeo9} confirms the approach is accurate also for a case where the intuitive arguments of Che \etal  do not apply. The population is only 100 and the popularity law is geometric with parameter $\rho=0.9$. The figure plots the hit rates of objects ranked 1, 4, 16 and 64. Discrepancies are visible only for object 1 and these are very slight.

%%%%%%%%%%%%%%%%%%%%%%%%%%%%%%%%%%%%%%%%%%%%%%%%%%%%%%%%%%%%%%%

\section{Why the approximation works}
\label{sec:whyitworks}
We explain why the Che approximation works so well, even for a small object population and a small cache size.

\subsection{A Gaussian approximation for $X(t)$}
Rather than attempting to study $T_C$ directly, we consider $X(t)$ and exploit the  elementary relation $\P(T_C>t)= \P(X(t)<C)$, for $t\geq 0$. Since  the variable $X(t)$ is a sum of independent random variables,  it is natural to expect its distribution to be approximately Gaussian.  %a central limit theorem and express $X(t)$ in terms of a Gaussian random variable. 

\begin{figure}
\     \vspace{-8mm}\hspace{-2cm}
  \subfloat[Zipf(0.8), $N=10000$]{\label{fig:norm-zipf-8}\includegraphics[width=85mm,height=55mm]{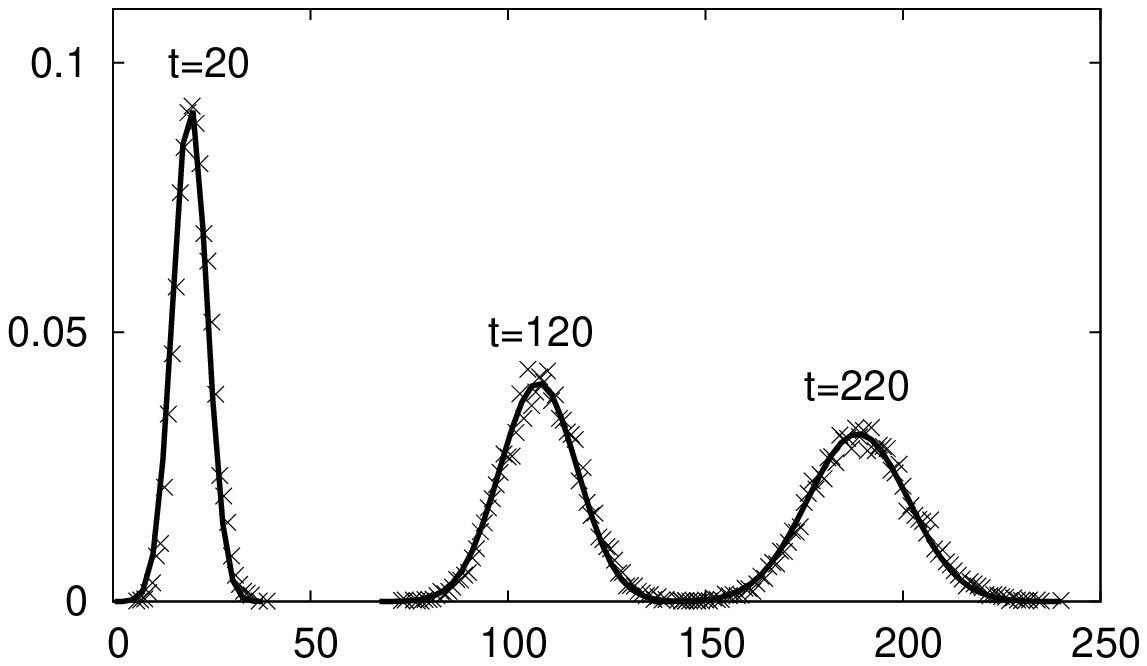}}  
%      \vspace{-5mm}
  \subfloat[Zipf(1.2), $N$=10000]{\label{fig:norm-zipf-12}\includegraphics[width=85mm,height=55mm]{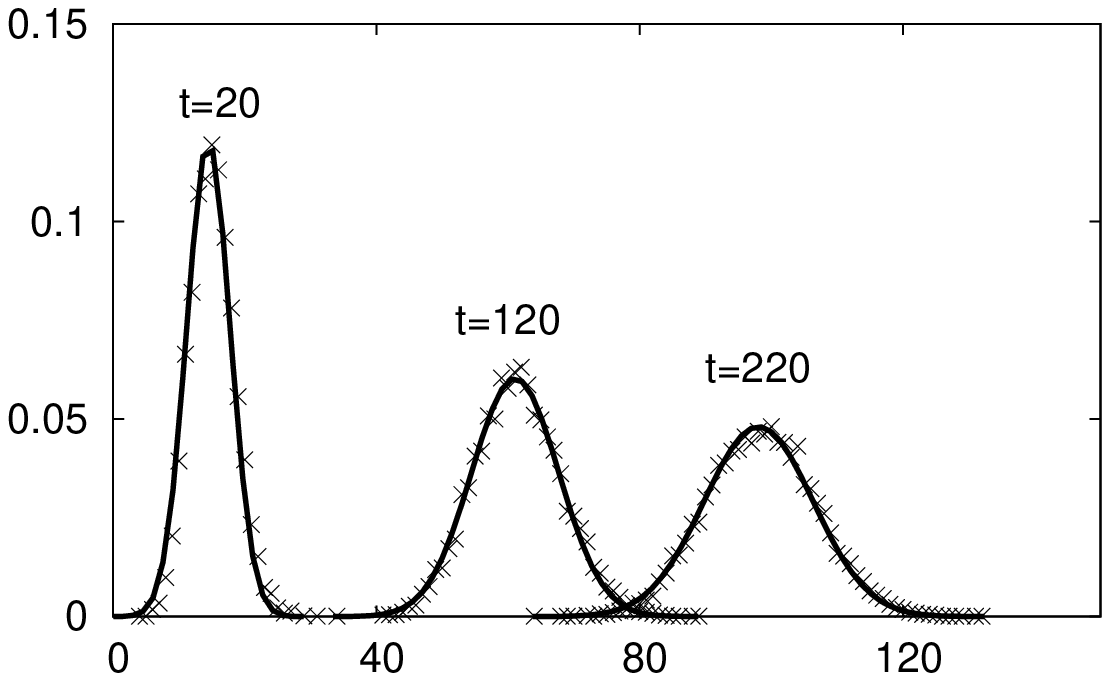}} \\ 
%      \vspace{-5mm}
\subfloat[Geo(0.9), $N$=100]{\label{fig:norm-geo-9}\includegraphics[width=85mm,height=55mm]{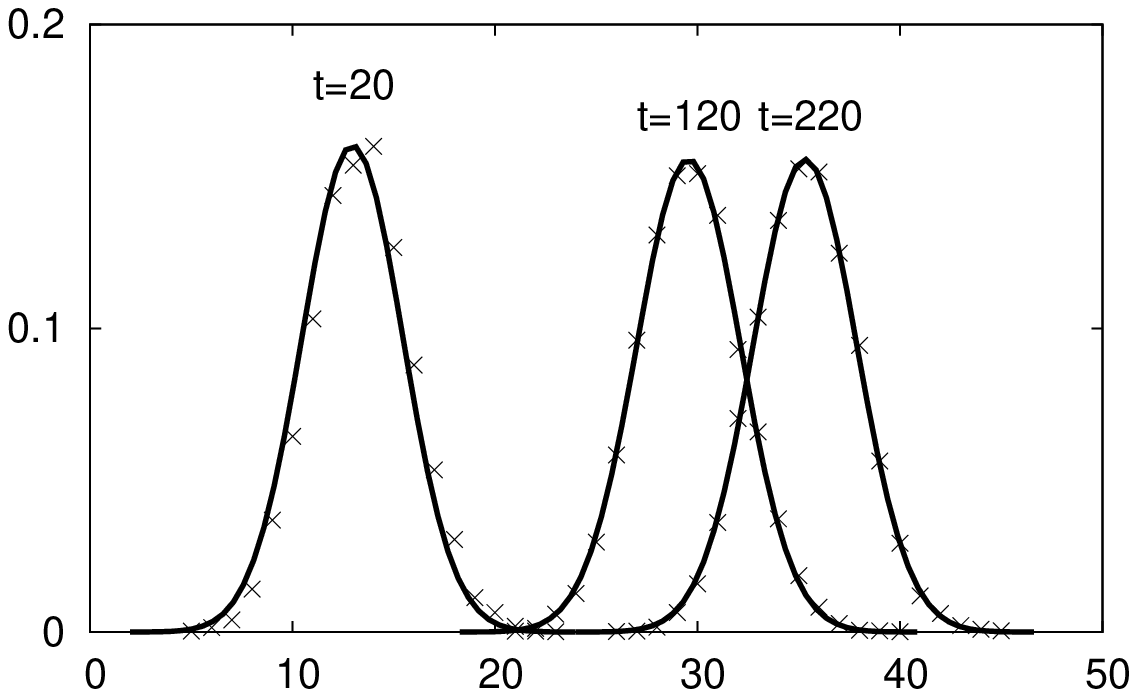}} 
   \caption{Distribution of $X(t)$, simulation and Gaussian approximation}
         \vspace{-5mm}

     \label{fig:normal}
\end{figure}

Figure \ref{fig:normal}  plots the distribution of $X(t)$ for some values of $t$ when the popularity distribution $q(n)$ is: \ref{fig:norm-zipf-8})  Zipf(.8) for $1\le n\le 10^4$; \ref{fig:norm-zipf-12}) Zipf(1.2) for $1\le n\le 10^4$; and \ref{fig:norm-geo-9}) Geo(.9) for $1\le n\le 100$. The figures plot simulation results as crosses with a superposed normal distribution of mean and variance given by (\ref{eq:m(t)}) and (\ref{eq:sigma2(t)}), respectively. The figures clearly suggest that $X(t)$ is indeed Gaussian for the range of popularity laws and times $t$ of interest. 

\subsection{A central limit theorem}

The following proposition establishes conditions under which a Gaussian approximation for $X(t)$ is reasonable. 

\begin{proposition}\label{Central} If 
$W(t){:=}\left(X(t)-m(t)\right)/\sigma(t)$, where $m(t)$ and  $\sigma(t)$ are given by (\ref{eq:m(t)}) and (\ref{eq:sigma2(t)}), respectively, then
\[
\|{\cal L}(W(t))-{\cal L}({\cal G})\|  {:=}\sup_{x\in\R} \left|\P(W(t)\leq x) {-}\P({\cal G}\leq x)\right| 
  {\leq}  \frac{K}{\sigma(t)},
\]
where ${\cal G}$ is a centered normal random variable and $K\le0.56$.
\end{proposition}

\begin{proof}
Let  $Z_n(t)=[\ind{\tau_n\leq t}- (1 - \exp(-q(n)t))]/\sigma(t)$ and
\[
Z(t) {:=} \frac{X(t)-m(t))}{\sigma(t)} = \sum_{n=1}^N Z_n(t).
\]
Then $\E(Z_n(t))=0$,  $\E(Z(t)^2)=1$ and we have,
\begin{multline*}
\sigma(t)^{3} \E\left(|Z_n(t)|^3\right){=} e^{-3q(n) t}(1{-}e^{-q(n) t}){+}(1{-}e^{-q(n) t})^3e^{-q(n) t}
\\\leq e^{-q(n) t}(1-e^{-q(n) t}).
\end{multline*}
Hence,
\[
\sum_{n=1}^{N} \E\left(|Z_n(t)|^3\right)\leq 
\frac{1}{\sigma(t)^{3}}  \sum_{n=1}^{N}e^{-q(n) t}(1-e^{-q(n) t})= \frac{1}{\sigma(t)}.
\]

Berry-Esseen's Inequality, see Feller~\cite[p.~544]{Feller}, gives  the relation
\[
\|{\cal L}(W(t))-{\cal L}({\cal G})\|\leq K \sum_{n=1}^{N}\E(|Z_n(t)|^3)\leq\frac{K}{\sigma(t)},
\]
where $K$ is a constant. It has recently been shown that the value of $K$ is no greater than 0.56 \cite{Shevstova2010}.
\end{proof}

This proposition confirms that $X(t)$ is asymptotically Gaussian as $t\to \infty$ if $\sigma(t)$ grows unboundedly. It is not wholly satisfactory, however, in that it does not explain the excellent fit illustrated in Figure \ref{fig:normal} for small values of $t$ and for geometric popularity where $\sigma(t)$ is not an increasing function. This appears to be a normal situation for central limit theorems where convergence to the normal distribution is often much better than predicted by the analytical bounds.

\subsection{Approximating the hit rates}
Starting with the Gaussian approximation for the distribution of $X(t)$, we argue that the Che approximation is indeed generally applicable.

\begin{proposition}
\label{prop:erfc}
Assuming $X(t)$ is Gaussian of mean $m(t)$ and standard deviation $\sigma(t)$, the hit rate for object $n$ may be written,
\begin{equation}
h(n) = 1 - \frac{1}{2}\int_0^{+\infty} \mathrm{erfc}\left(\frac{C-m(u)}{\sqrt{2}\sigma(u)}\right) qe^{-q u}\,du
\label{eq:erfc}
\end{equation}
where $\mathrm{erfc}(x)$ is the complementary error function.
\end{proposition}

\begin{proof}
Recall that $h(n)=1-\E\left(e^{-q(n) T_C}\right)$. Consider $\E\left(e^{-q T_C}\right)$ for some $q>0$.
By definition of $T_C$,
\[
\E\left(e^{-q T_C}\right)=\int_0^{+\infty} \P(T_C<u) qe^{-q u}\,du=\int_0^{+\infty} \P(X(u)>C) qe^{-q u}\,du
\]

By the assumption that $X(t)$ is Gaussian, we can write
$$\P(X(u)>C) = \frac{1}{2}\mathrm{erfc}\left(\frac{C-m(u)}{\sqrt{2}\sigma(u)}\right) $$
and the proposition follows.
\end{proof}

\begin{figure}
\centering
\vspace{-10mm}
\hspace{-15mm}
{%\large
\resizebox{10cm}{!}{\input{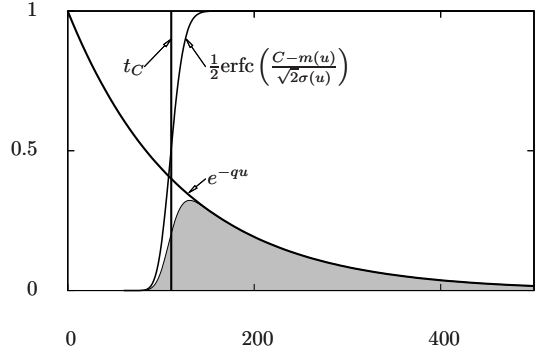}}
}
\vspace{-5mm}
\caption{Approximating the integral}
\label{fig:meanval}
\end{figure}

Proposition \ref{prop:erfc} could be used directly to evaluate the hit rates. Instead, we derive the Che approximation as an approximation for the integral in (\ref{eq:erfc}). Note that the complementary error function in the integrand is an S-shaped function tending rapidly to asymptotes at 0 and 1 from a point of inflection at $m(u)=C$, i.e., at $u=t_C$. We therefore replace this function in the integral by the step function $ \ind{m(u)>C}$ yielding the following approximation
\begin{align*}
\E\left(e^{-q T_C}\right) & \approx  \int_0^{+\infty} \ind{m(u)>C} qe^{-q u}\,du\\
                                          & =   \int_0^{+\infty} \ind{u>t_C} qe^{-q u}\,du 
                                           =  e^{-q t_C}.
 \end{align*}
The second step follows from $m(t_C)=C$ and the fact that $m(u)$ is increasing in $u$. This establishes the validity of the Che approximation on condition that replacing erfc by a step function is accurate.

This accuracy is illustrated in Figure  \ref{fig:meanval} for a particular set of parameters: Zipf(0.8) popularity, $N=10000$, $C=100$ and $q=.83 \times 10^{-3}$. The exact integral is the shaded area. The approximation replaces this by the area to the right of $t_C$ and under the exponential curve. Visibly, the approximation is good in this case. It is easy to convince oneself that this is generally true for the popularity laws of interest. However, it unfortunately does not seem possible to quantify the error due to the non-explicit nature of the erfc function argument in Proposition \ref{prop:erfc}.

\section{Zipf popularity and a large cache}
\label{sec:largecache}

For Zipf law popularity, we can prove the asymptotic validity of the Che approximation and characterize $t_C$ directly, without solving equation (\ref{eq:hitC}). %Note first that for Zipf popularity, from (\ref{eq:m(t)}) and (\ref{eq:sigma2(t)}),  $m(t) \to \infty$ and $\sigma(t) \to \infty$ as $N \to \infty$.

\subsection{Preliminary results}
We first prove two lemmas about the moments of $X(t)$.

\begin{lemma}\label{lem:estvar} For $t\geq 0$, the variance of $X(t)$ can be expressed in terms of its mean:
\[
\sigma(t)^2=m(2t) - m(t).
\]
\end{lemma}

\begin{proof}
From (\ref{eq:sigma2(t)}),
\[
\sigma(t)^2= \sum_{n=1}^{N} e^{-q(n) t}\left(1-e^{-q(n) t}\right)
= \sum_{n=1}^{N} \left(1- e^{-2q(n) t}\right) - \left(1-e^{-q(n) t}\right),
\]
which yields the desired identity. 
\end{proof}

Now consider the behaviour of the average of $(X(\cdot))$ at an appropriate time scale.

\begin{lemma}\label{lem:AsympE} 
With $q(n)=1/n^{\alpha}$ for $1\le n \le N$, for any $\beta>0$, 
\[
\E(X(\beta N^\alpha ))= \psi_\alpha(\beta)N+o(N),
\]
where 
\[
\psi_\alpha(\beta)= 1-\int_0^1 e^{-\beta/x^\alpha}\,dx.
\]
\end{lemma}

\begin{proof}

\[
\E(X(\beta N^\alpha ))= N\times \frac{1}{N}\sum_{n=1}^N \left(1-e^{-\beta N^\alpha/n^\alpha}\right)=
N\left(\int_0^1 1-e^{-\beta/x^\alpha}\,dx\right) +o(N).
\]

\end{proof}

\subsection{A Gaussian approximation for $T_C$}
In the particular case of Zipf law popularity we show that $T_C$ is asymptotically Gaussian and derive $t_C$ as its expectation. In the following $\lfloor z\rfloor$ denotes the integer part of $z\geq 0$.

\begin{proposition}
\label{tclim}

For Zipf($\alpha$) popularity, as $N$ and $C$ tend to infinity with $C=\lfloor \delta N\rfloor$, for $0<\delta<1$, we have
\begin{equation}
t_{\lfloor \delta N\rfloor}=\psi_\alpha^{-1}(\delta) N^{\alpha}+o(N^{\alpha}),
\label{eq:tC}
\end{equation}
where $\psi_\alpha(\beta)$ is defined in Lemma \ref{lem:AsympE}.

Furthermore, the random variable 
\[
\frac{\psi_\alpha'(\psi_\alpha^{-1}(\delta))}{\sqrt{\psi_\alpha(2\psi_\alpha^{-1}(\delta))-\delta}}\frac{(T_{\lfloor \delta N\rfloor}-t_{\lfloor \delta N\rfloor})}{N^{\alpha-1/2}}
\]
converges in distribution to a centred Gaussian random variable. 
\end{proposition}

\begin{proof}
The asymptotic relation for $t_C$ is a direct consequence of  Lemma \ref{lem:AsympE}. We have,
$$\E(X(\beta N^{\alpha}))= \psi_{\alpha}(\beta)N+o(N)$$ 
and, by definition of $t_C$,
$$E(X(t_c))=\lfloor \delta N \rfloor = \delta N + o(N).$$

Let $s_C=\psi_a^{-1}(\delta)N^{\alpha}$ so that 
$$\E(X(s_C))=\psi_{\alpha}(\psi_{\alpha}^{-1}(\delta)) N +o(N)=\delta N + o(N).$$
It follows that $t_C =s_C+o(N^{\alpha}) $ and the first statement of the proposition is proved.

Now consider the following relation, for $x\in\R$,
\begin{multline*}
\P(T_C-t_C\geq x)=\P(X(t_C+x)<C )=\\\P\left[\frac{X(t_C{+}x){-}m(t_C{+}x)}{\sigma(t_C{+}x)}{<}\frac{m(t_C){-}m(t_C{+}x)}{\sqrt{m(2(t_C{+}x)){-}m(t_C{+}x)}}\right].
\end{multline*}
The second equality follows on setting $m(t_C)=C$ and applying Lemma~\ref{lem:estvar}. 

Since for the Zipf laws, $\sigma(\beta N^\alpha)$ goes to infinity as $N\to \infty$, Proposition~\ref{Central} shows that the variable
\[
\frac{X(\beta N^\alpha)-m(\beta N^\alpha)}{\sigma(\beta N^\alpha)}
\]
is arbitrarily close to a centred normal variable if $N$ is sufficiently large. The only thing to check is the asymptotic behaviour of the fraction
\[
\frac{m(t_C)-m(t_C+x)}{\sqrt{m(2(t_C+x))-m(t_C+x)}}
\]
when $x=\beta N^{\alpha-1/2}$. 

Using Lemma \ref{lem:AsympE}, write $m(t_C+x)$ as
\begin{eqnarray*}
m(t_C+x) & = & m\left(N^{\alpha}(\psi_{\alpha}^{-1}(\delta)+\frac{x}{N^{\alpha}}+o(1))\right) \\
                 & = & N \psi_{\alpha}\left(\psi_{\alpha}^{-1}(\delta)+\frac{x}{N^{\alpha}} \right) + o(N).
\end{eqnarray*}
Expanding $\psi_{\alpha}$ about $\psi_{\alpha}^{-1}(\delta)$ yields the numerator 
$$ m(t_C) - m(t_C+x) = - \psi'_{\alpha}(\psi_{\alpha}^{-1})(\delta)N^{1-\alpha} x + o(N).$$
The denominator can similarly be written
$$\sqrt{m(2t_c+x))-m(t_C+x)} = \sqrt{\psi_{\alpha}(2\psi_{\alpha}^{-1}(\delta))N -\delta N+ o(N)}.$$
The second statement of the proposition follows on substituting for $x$. 

% This is a straightforward application of Lemma~\ref{lem:AsympE}. 
\end{proof}

\vspace{2mm}
\noindent {\bf Remark.} The expression for $t_{\lfloor \delta N\rfloor}$ in Proposition \ref{tclim} coincides with an equivalent quantity derived differently in Theorems 2 and 3 of Jelenkovic \etal   \cite{JKR05} for $\alpha=1$ and $\alpha<1$, respectively. Theorem 1 of the same paper provides an explicit expression for $t_C$ when $\alpha>1$ and $N$ is infinite. This expression proves significantly less precise than (\ref{eq:tC}) when $\alpha$ is not much greater than 1 (1.2, say), even for $N$ as large as 10000.

\subsection{The Che approximation}
Proposition \ref{tclim} shows that there is a function $\Theta(N)$ such that, as $N\to \infty$, $\Theta(N)/t_C \to0$ while $(T_C-t_C)/\Theta(N)$ converges to a Gaussian random variable. Thus in this case, $T_C$ does indeed become deterministic (i.e., $T_C/t_C \sim 1$) and the original argument of Che \etal applies. We have $\E (e^{-qT_C}) \to e^{-qt_C}$ as $C$ (and $N$) $\to \infty$. 

\subsection{Geometric popularity}
It can be shown for geometric popularity, $q(n)=\rho^n$ for $n\ge 0$, that $m(t)=-\log t/ \log \rho+O(1)$. Thus, by Lemma \ref{lem:estvar}, $\sigma(t)^2 = \log 2/ \log \rho +O(1)$, i.e., the variance of $X(t)$ is asymptotically constant and small compared to $m(t)$. This explains why the Che approximation works for geometric popularity (applying the arguments in Section \ref{sec:whyitworks}) although $T_C$ is by no means deterministic.

\section{A ``Che approximation'' for random replacement}
\label{sec:randomche}
The excellent accuracy of the Che approximation for LRU replacement motivates the search for a similar approach for other policies. In this section we consider random replacement and derive an approximation that is similar in accuracy and complexity to the Che approximation.

Random might be preferred to LRU because it is simpler to implement. When a new object is to be added to the cache, it overwrites a randomly chosen existing object independently of the popularity of the latter. This policy was shown by Gelenbe to have exactly the same hit rates as FIFO  \cite{Gelebe73}.  

The exact analysis of Gelenbe is too complex for practical evaluation. A recent paper by Simonian \etal \cite{KS12} provides large cache asymptotics applicable for a Zipf popularity law with $\alpha>1$. Dan and Towsley  \cite{DT90} propose an approximate evaluation for the hit rates of a FIFO cache.  %It has the same complexity as the approximation proposed by Dan and Towsley.

Note that $h(n)$, the hit rate for object $n$, is the probability object $n$ is in the cache at an arbitrary instant. This can be expressed by Little's formula as the product $\lambda(n)\times T(n)$ where $\lambda(n)$ is the frequency at which object $n$ enters the cache and $T(n)$ is its average sojourn time. 

Clearly, $\lambda(n) = (1-h(n))q(n)$.  We assume $T(n)$ is inversely proportional to the arrival rate of requests for objects other than $n$. This is only approximately true but is intuitively reasonable. We deduce, $h(n) = (1-h(n))q(n) \times  \tau_C/\sum_{i\ne n}q(i)$ or,
\begin{equation}
h(n) = \frac{q(n) \tau_C}{\sum_{i\ne n}q(i)+q(n)\tau_C},
\label{eq:h(n)rand}
\end{equation}
 for some unknown constant $\tau_C$. Equating the sum of hit rates to the cache size $C$, as in Section \ref{sec:cheapprox}, yields the equation for $\tau_C$,
\begin{equation}
C= \sum_{n=1}^N \frac{q(n) \tau_C}{\sum_{i\ne n}q(i)+q(n)\tau_C}.
 \label{eq:cherand}
 \end{equation}
 Equation (\ref{eq:cherand}) is the `random' equivalent to the LRU Che identity (\ref{eq:hitC}). Solving for $\tau_C$ yields the hit rates via (\ref{eq:h(n)rand}).

Figure \ref{fig:rand-approx} shows results analogous to those of Figure \ref{fig:approx} for LRU. The accuracy is  clearly comparable. It largely remains to analyse why this is so but note that the assumption sojourn times are proportional to the request rate of other objects appears equally reasonable for all realistic popularity laws. The FIFO algorithm of Dan and Towsley \cite{DT90} is similarly accurate. 
 
\begin{figure}
  \subfloat[$N=10^4$, Zipf(0.8), objects 1, 10, 100, 1000]{\label{fig:gull}\includegraphics[width=80mm,height=50mm]{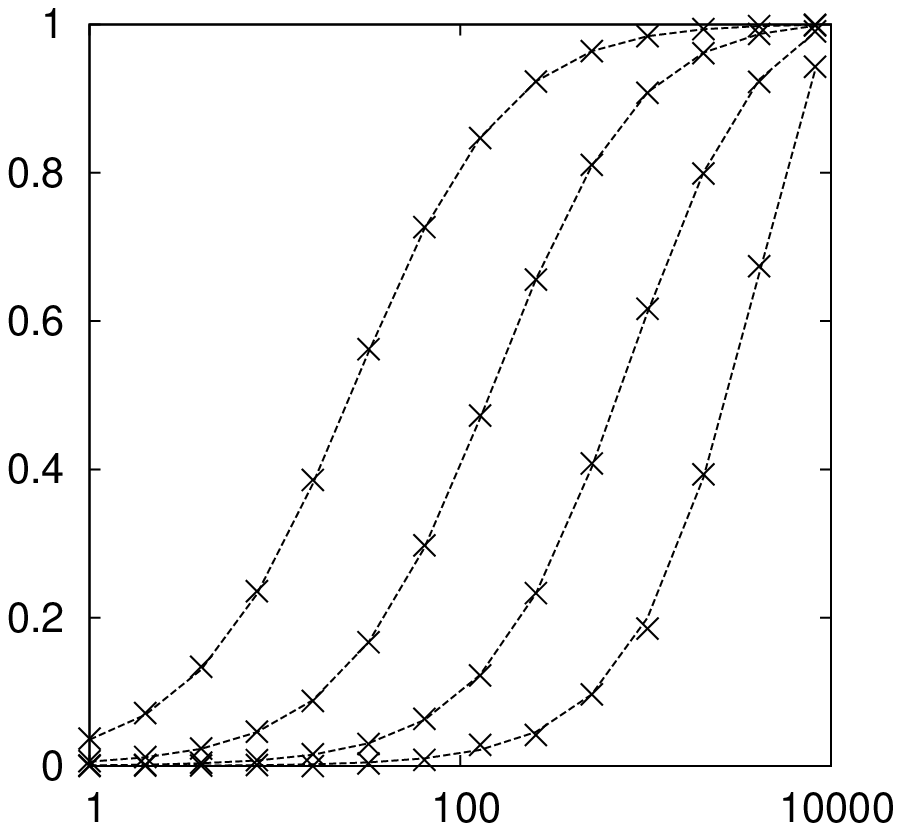}}   
%\hspace{-15mm}     
  \subfloat[$N=10^4$, Zipf(1.2), objects 1, 10, 100, 1000]{\label{fig:tiger}\includegraphics[width=80mm,height=50mm]{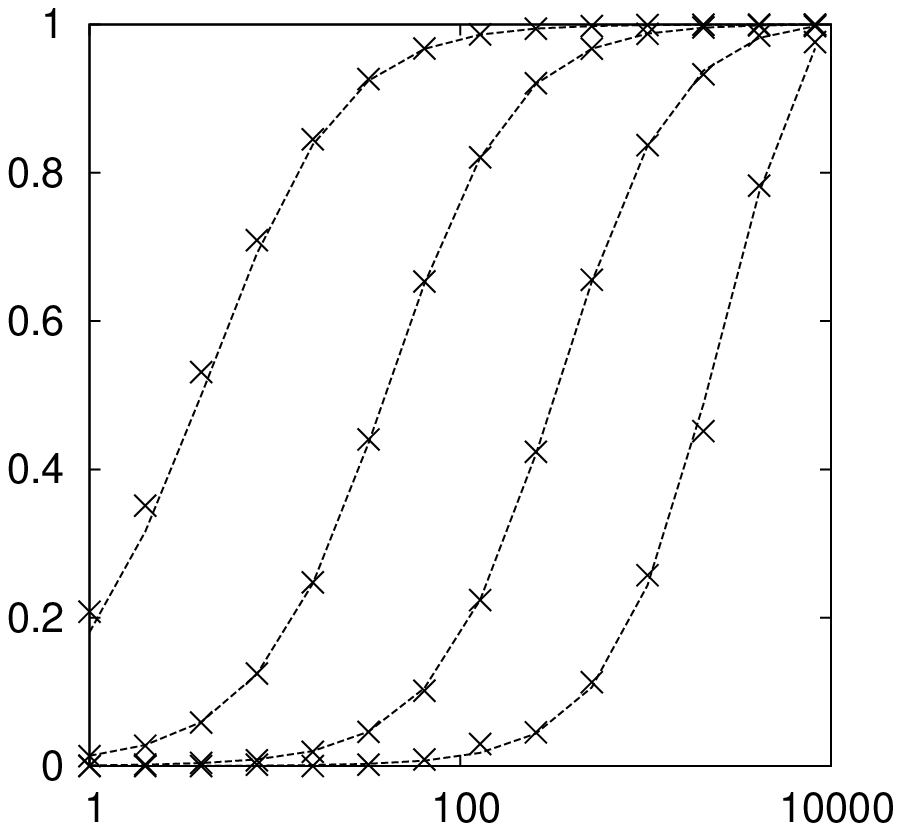}} \\
\hspace{-15mm}     
\subfloat[$N=100$, Geo(0.9), objects 1, 4, 16, 64]{\label{fig:mouse}\includegraphics[width=80mm,height=50mm]{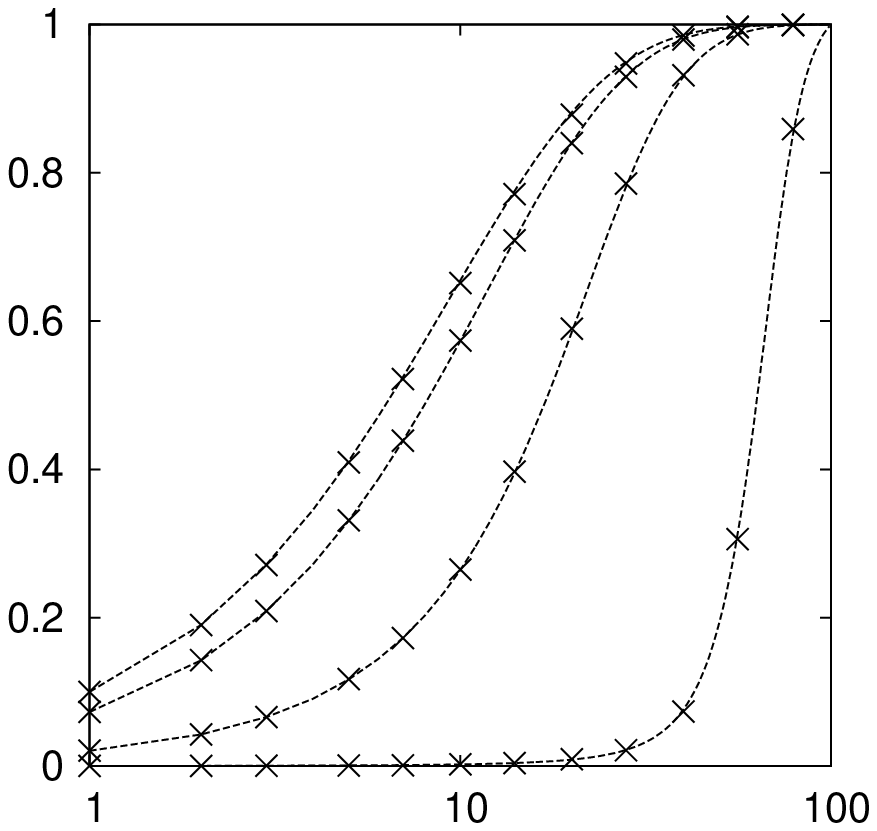}} 
   \caption{Hit rate against cache size for selected objects, random replacement}
  \label{fig:rand-approx}
\end{figure}

\section{Application}
\label{sec:application}

In this section we present an application that is intended to illustrate the power of the Che approximation. We revisit the networking example introduced in \cite{FRRS12} where users retrieve a mixture of web, file sharing, user-generated content (UGC) and video-on-demand (VoD) content via a cache.

\begin{table}[t]
\begin{center}
\begin{tabular}{l  |  c |c| c| c }
\hline
 & traffic share &  population  & object size & popularity \\
  &($p_i$)  & ($N_i$)&  ($\theta_i$) & ($\alpha_i$)\\
   \hline
Web   & .18 & $10^{11}$ & 10  & 0.8 \\
File sharing  & .36 & $10^{5}$  &  $10^6$ & 0.8\\
UGC  & .23 &  $10^8$  &  $10^3$ & 0.8 \\
VoD & .23 & $10^4$ &  $10^4$ & 1.2 \\ 
\hline
\end{tabular}
\caption{ Internet content traffic characteristics}
\label{tab:characteristics}
\vspace{-3mm}
\end{center}
\end{table}

Traffic characteristics and assumed popularity laws are presented in Table \ref{tab:characteristics}. Note the very large populations and diverse popularity laws. These make other performance evaluation approaches, including simulation, impractical. 

We suppose objects are divided into 1 KB chunks. For the sake of simplicity, objects of the same type $i$ are supposed to have constant size $\theta_i$ chunks.  We assume objects have Zipf popularity with exponent $\alpha_i$ and chunks inherit the popularity of their parent object.  The proportion of type $i$ traffic in bit/s downloaded by users is $p_i$. 

Given these assumptions we deduce the popularity of chunk $k$ of object $n$ of type $i$, for $1 \le i \le 4$, $1\le n \le N_i$ and $1\le k \le \theta_i$,
$$ q_i(n,k) = \frac{p_i/n^{\alpha_i} }{ \sum_{j=1}^{N_i} \theta_i/j^{\alpha_i} }.$$

In applying the Che approximation we optimize summations over as many as  $10^{11}$ objects in (\ref{eq:hitC}) and (\ref{eq:cherand}) by grouping successive terms which are nearly equal. Computation is then very rapid.

\begin{figure}
  \centering
\includegraphics[width=0.8\textwidth]{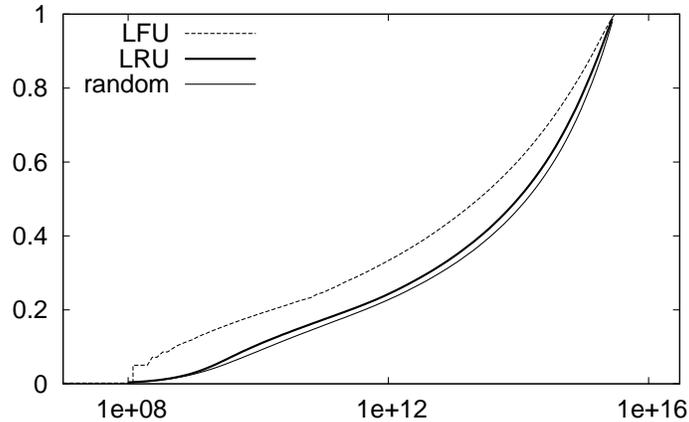}    
\caption{Hit rate against cache size (in bytes) for the traffic mix of Table \ref{tab:characteristics} with LFU, LRU and random replacement}
 \label{fig:lruvsrand}
 \vspace{-2mm}

\end{figure}

Figure \ref{fig:lruvsrand} compares the overall hit rate for the traffic mix as a function of cache size for three different replacement policies: least frequently used (LFU),  LRU and random. The LFU hit rate is calculated as in \cite{FRRS12} while for LRU and random we use the Che approximations of Sections \ref{sec:cheapprox} and \ref{sec:randomche}, respectively. 

The significance of these and similar results is discussed in \cite{FRRS12}. An additional observation is that random is hardly worse than LRU in this case. Our main objective in presenting these results is to stress that they are readily derived using the Che approximation when, in view of the huge populations and diversity of content objects, any other approach would be impracticable or inexact.

 \section{Conclusion}

The Che approximation constitutes a versatile and highly accurate tool for predicting the hit rate performance of a cache with LRU replacement. We have demonstrated in the paper why the approximation works so well, even when the conditions suggested by its authors are not satisfied. The analysis lends confidence to using this tool to evaluate the performance of an information-centric network where the large populations and diversity of content catalogues preclude utilization of alternative approaches. Note, in particular, that the Che approximation can be usefully combined with the approach in \cite{RKT2010} to evaluate large-scale, general cache networks.

\end{document}